\documentclass[a4paper,12pt]{article}
\usepackage{cite}
\usepackage{bussproofs, amsmath,amsfonts,amssymb,amsthm,epsfig,epstopdf,titling,url,array,listings}
\usepackage{etoolbox,refcount}
\usepackage{amsthm}
\usepackage{bbm}
\usepackage{multicol}
\usepackage{indentfirst}
\usepackage{mathtools}
\usepackage[linguistics]{forest}

\newtheorem{thm}{Theorem}

\theoremstyle{plain}
\newtheorem{theorem}[thm]{Theorem}

\theoremstyle{plain}
\newtheorem{lemma}[thm]{Lemma}

\theoremstyle{definition}
\newtheorem{definition}[thm]{Definition}

\theoremstyle{remark}
\newtheorem*{remark}{Remark}

\title{Canonicity and Computability in Homotopy Type Theory}
\author{Author: Dmitry Filippov
\and 
Supervisor: Kobi Kremnitzer}

\newcounter{countitems}
\newcounter{nextitemizecount}
\newcommand{\setupcountitems}{%
  \stepcounter{nextitemizecount}%
  \setcounter{countitems}{0}%
  \preto\item{\stepcounter{countitems}}%
}
\makeatletter
\newcommand{\computecountitems}{%
  \edef\@currentlabel{\number\c@countitems}%
  \label{countitems@\number\numexpr\value{nextitemizecount}-1\relax}%
}
\newcommand{\nextitemizecount}{%
  \getrefnumber{countitems@\number\c@nextitemizecount}%
}
\newcommand{\previtemizecount}{%
  \getrefnumber{countitems@\number\numexpr\value{nextitemizecount}-1\relax}%
}
\makeatother    
{\end{itemize}%
\unskip\computecountitems\ifnumcomp{\previtemizecount}{>}{3}{\end{multicols}}{}}

\makeatletter
\DeclareFontFamily{OMX}{MnSymbolE}{}
\DeclareSymbolFont{MnLargeSymbols}{OMX}{MnSymbolE}{m}{n}
\SetSymbolFont{MnLargeSymbols}{bold}{OMX}{MnSymbolE}{b}{n}
\DeclareFontShape{OMX}{MnSymbolE}{m}{n}{
    <-6>  MnSymbolE5
   <6-7>  MnSymbolE6
   <7-8>  MnSymbolE7
   <8-9>  MnSymbolE8
   <9-10> MnSymbolE9
  <10-12> MnSymbolE10
  <12->   MnSymbolE12
}{}
\DeclareFontShape{OMX}{MnSymbolE}{b}{n}{
    <-6>  MnSymbolE-Bold5
   <6-7>  MnSymbolE-Bold6
   <7-8>  MnSymbolE-Bold7
   <8-9>  MnSymbolE-Bold8
   <9-10> MnSymbolE-Bold9
  <10-12> MnSymbolE-Bold10
  <12->   MnSymbolE-Bold12
}{}

\let\llangle\@undefined
\let\rrangle\@undefined

\DeclareMathDelimiter{\llangle}{\mathopen}%
                     {MnLargeSymbols}{'164}{MnLargeSymbols}{'164}
\DeclareMathDelimiter{\rrangle}{\mathclose}%
                     {MnLargeSymbols}{'171}{MnLargeSymbols}{'171}
\makeatother

\begin{document}

\maketitle

\begin{abstract}
This dissertation gives an overview of Martin L\"{o}f's dependant type theory, focusing on it's 
computational content.

\end{abstract}

\section{Introduction}

Type theories are a powerful mathematical tool that has been used for a long time in Computer Science for the 
purpose of mandating the safe application of functions by limiting their scope to only objects for which their 
use was intended. In this sense, their use in programming languages is similar to their original intent of 
limiting the scope of mathematical constructions in order to prevent a series of paradoxes stemming from 
self-reference. 

Recently the interest in mathematical applications of type theory has grown thanks to type theory's intimate 
connection to constructive mathematics and type theory. It has been utilised in proof assistant software and 
formal proof verification, as well as used as a foundation for constructive mathematics in Univalent Foundations 
Project.

This dissertation is mostly based on the HoTT book \cite{UFP13} and attempts to discuss computational content 
of Univalence axiom and it's potential judgemental handling, in particular presenting the result from \cite{Licata12}.

\section{Constructive Mathematics}

Constructive mathematics is a different philosophical approach to the subject that postulates that mathematical 
objects exist only in so much as us being able to construct and use them, in contrast to classical mathematics 
that presupposes an existing mathematical reality that a mathematician merely discovers. 

These differences manifest in a different handling of truth. Classical mathematics sees the truth as an inherent 
properties of mathematical reality. As such, every proposition is assumed to either be true or to be false 
(Law of Excluded Middle) and proves by contradiction are allowed. In constructive mathematics, on the other hand, 
LEM is rejected; it is thought that a proof by contradiction, while showing that a proposition cannot be 
disproved, does not actually provide an evidence for said proposition. This is most apparent in a case of 
existential statements - a proof of such statement by contradiction does not give a way to construct a witness 
for this statement, and since constructivism equates an existence of a mathematical structure with a possibility 
of it's construction, such argument does not prove existence.

\subsection{Constructive Logic}

Constructive logic can be presented as a Hilbert-style calculus, similar to classical logic, but with LEM removed:

\begin{definition}[Constructive firt order logic]

Axioms:
\begin{itemize}
    \item $\rightarrow 1$: $\phi \rightarrow \chi \rightarrow \phi$
    \item $\rightarrow 2$: $(\phi \rightarrow \chi \rightarrow \psi) \rightarrow ((\phi \rightarrow \chi) \rightarrow \phi \rightarrow \psi)$
    \item $\wedge 1$: $(\phi \wedge \chi) \rightarrow \phi$
    \item $\wedge 2$: $(\phi \wedge \chi) \rightarrow \chi$
    \item $\wedge 3$: $\phi \rightarrow \chi \rightarrow (\phi \wedge \chi)$
    \item $\vee 1$: $\phi \rightarrow (\phi \vee \chi)$
    \item $\vee 2$: $\chi \rightarrow (\phi \vee \chi)$
    \item $\vee 3$: $(\phi \rightarrow \psi) \rightarrow (\chi \rightarrow \psi) \rightarrow (\phi \vee \chi) \rightarrow \psi$
    \item $\bot 1$: $\bot \rightarrow \phi$
\end{itemize}

A single rule of inference is 
\begin{itemize}
    \item Modus Ponens : from $\phi$ and $\phi \rightarrow \psi$ deduce $\psi$
\end{itemize}

This can be extended to a predicate logic with

Following axioms:

\begin{itemize}
    \item $\forall 1$: $(\forall x \phi(x)) \rightarrow \phi(t)$
    \item $\exists 1$: $\phi(t) \rightarrow \exists x \phi (x)$
\end{itemize}

And rules of inference:

\begin{itemize}
    \item from $\psi \rightarrow \phi$ deduce $\psi \rightarrow \forall x \phi$ if $x \notin FV(\psi)$
    \item from $\psi \rightarrow \phi$ deduce $(\exists x \psi) \rightarrow \phi$ if $x \notin FV(\phi)$
\end{itemize}

\end{definition}

A more natural presentation of constructive logic, however, is given in a style of natural deduction \cite{Max07}, with 
each propositional connective given introduction and elimination rules. The reason such representation is more 
natural is due to the connection between constructive logic and Type Theory which we explore below.

The Hilbert style presentation of constructive logic can also be seen as having introduction ($\vee 3$)
and elimination ($\vee 1$, $\vee 2$) rules. From this perspective, implication is being used as a function 
symbol, "eliminating" the inputs to "introduce" the outputs. This is exactly why the axioms for implication 
($\rightarrow 1$, $\rightarrow 2$) have the same shape as combinators in $\lambda$-calculus.

\subsection{Type Theory}

Type Theory is a framework of studying mathematics where every object is given a type which governs which 
properties we expect the object to have and which operations we expect to be able to apply to it, similar to a 
type in a programming language. 

There are two flavours of Type Theories-  extrinsic and intrinsic. In an 
extrinsic type theory, one sees type as merely an assignment of an object to a type; such a theory sees an 
identity function $\lambda x.x$ as the same function when applied to any type, be it an integer or a boolean 
with type being attributed to it at the moment of application as a type derivation of an expression. In an 
intrinsic Type Theory, $\lambda x.x: Bool \rightarrow Bool$ and $\lambda x.x: Int \rightarrow Int$ are two 
different functions, assigned type on construction. In this Dissertation we are concerned with intrinsic Type 
Theories.

In intrinsic type theory, every object is constructed as an element of a particular type using the 
\textbf{constructors} of that type and has a defined set of ways it can be used defined as 
\textbf{eliminators} of that type. Thanks to this explicitly computational behaviour, intrinsic type theories 
are easier to implement on a computer, since instead of functions being defined as infinite relations as they 
are in Set Theory, functions are defined as as procedure using eliminators of a type more akin to how one 
would define a program.

$\lambda$-calculus is the prototypical example of type theory, encoding general computable functions as a 
substitution calculus with function abstraction and application as basic constructor and eliminator.

\subsection{Curry-Howard Correspondence}

Since constructive mathematics treats truth differently, there are different semantics for what it means for 
a proposition to be true; a truthfulness of a proposition is equated with it's proof and a proof of a proposition is
defined by the Brouwer–Heyting–Kolmogorov (BHK) interpretation:
\begin{itemize}
    \item A proof of $P \wedge Q$ is an ordered pair $\langle a, b \rangle$ consisting of a proof a of P and proof b of Q
    \item A proof of $P \vee Q$ is either $\langle 0, a \rangle$ where a is a proof of P or $\langle 1, b \rangle$ where b is a proof of Q
    \item A proof of $P \rightarrow Q$ is a function or a procedure f that converts a proof of P into a proof of Q
    \item A proof of $(\exists x \in S)P$ is a pair $\langle x, a \rangle$ where x is an element of S and a is a proof of P using x
    \item A proof of $(\forall x \in S)P$ is a function F that converts an element x of S into a proof of P
    \item $\bot$ is the absurdity, proposition with no proofs
    \item The formula $¬P$ is defined as $P \rightarrow \bot$ 
\end{itemize}

With this, it is more apparent how constructive logic is related to Type Theory - we can treat each proposition 
as a type of it's proofs, where ways of constructing a proof of a proposition given proofs of other propositions 
are seen as type cosntructors. The truth of a proposition then corresponds to the inhabitability of the 
corresponding type. 

This correspondence was first noted by Curry and Howard and is hence called the Curry-Howard correspondence
\cite{wadler2015propositions}. This correspondence shows that the rules for construction 
and elimination of terms in Type Theories correspond exactly to a natural deduction representation of 
constructive logics, where a BHK interpretation of proof of propositions is expressed as membership of a 
particular type, corresponding to a proposition.

Example of correspondences between type theories and constructive logics:

\begin{itemize}
    \item implicative part of propositional logic is represented by $\lambda$-calculus
    \item proving proposition $\alpha$ from assumptions $\Gamma$ ($ \Gamma \vdash \alpha $) is the same as writing a program that 
given objects of types $\Gamma$ constructs an object of type $\alpha$
    \item logical axioms correspond to rules of introducing a new variable with an unconstrained type
    \item $\rightarrow$introduction rule corresponds to $\lambda$-abstraction
    \item $\rightarrow$elimination rule (Modus Ponens) corresponds to function applications
    \item provability corresponds to type inhabitability
    \item constructive tautology corresponds to inhibitable type
\end{itemize}

Moreover, Type Theory can be used as a semantic domain for the constructive mathematics; the way Type Theories
are defined as a prototypical programming language lends itself nicely to this purpose. As such, Type Theory 
can serve as both the deductive system and semantic model for constructive mathematics, as opposed to 
classical mathematics that needs separate logical and semantic models (first order logic and set theory 
respectively).

\subsection{Dependant Types and Predicates}
Regular $\lambda$-calculus corresponds to implicative part of constructive logic, defining the behaviour of implication 
as a function. Other propositional connectives can be defined as simple types, largely imitating the introduction 
and elimination rules from Natural Deduction.  In order to talk about predicates, however, we need a radically 
different sort of construction.
 
Predicates in logic can be considered as a map between a collection of objects and a collection of propositions.
In constructive logic a predicate maps a type (seen as a collection of objects) to a collection of types 
(seen as propositions). This corresponds to an idea of dependent types - a type family indexed by the 
elements of some other type. For example, a family of types corresponding to cyclic groups of a given order: 

\[ n : \mathbb{N} \vdash Cyclic(n) Type \] 

Usually, such type families are made by a type construction with a free variable of the index type, just like 
predicates in logic are often formulae with a free variable.

Logical quantifiers in the predicate logic bind these free variables by "quantifying" over them. In Type Theory, 
these quantifiers correspond to dependant types which take a type family indexed over a type and return a type 
dependant on one less variable. Under BHK interpretation existential quantification is done by a dependant 
pair and universal quantification is done by the dependant functions. 

\subsection{Equality}
Since we interpret every proposition as a type, propositional equality also
needs to be a type. This is a ditinctive feature of Type Theory that we will
talk more about in next Section.

\section{Martin L\"{o}f Type Theory}

Martin L\"{o}f's theory of types is a more complex type theory than $\lambda$-calculus. It incorporates both 
dependant types and identity types, thus completely expressing the first order constructive logic.

\subsection{Intentionally}

One big deviation of Type Theory from Set Theory is Intentionally. Set Theory is extensional in a sense that 
if two objects (sets) contain the same elements, then they are identical. This principle is 
not true in Type Theory; in Type Theory every type and every element of a type has an intent - a specific 
token it is associated with by construction.

One way of illustrating this is to note that a collection of even divisors of 4 and a collection of even 
numbers smaller than 5 are the same set in Set Theory, but are different types in the Type Theory, since they 
are constructed differently, with a different intent and possess different computational content. This separation 
is particularly useful for polymorphic constructions - a function from a type of even divesers of a number can 
use the computational content present in elements of such type which would be absent from a second type.

This is reminiscent of the way datatypes are treated in Programming. Two  
datatypes that carry the same information might nevertheless be different and might have different overhead 
structure implemented on them; as such, a function that accepts one of these datatypes as an input might 
nevertheless reject the other. 

The Intentionality of Type Theory extends to the elements as well. Martin L\"{o}f Type Theory is intrinsic; 
that is, a type is assigned to the 
token on creation and any single token can only belong to a single type. This is akin to strongly typed 
languages, where a given variable can only have one type which cannot change. This, again, insures that the 
constructions such as functions are only applied to the tokens they are designed for. 

To conclude, the way Set Theory is designed is more akin to early assembly languages, where all data is treated 
as a sequence of 0s and 1s and can be interpreted as an encoding of any datatype post-factum. Intentional Type 
Theory on the other hand is more akin to strongly-typed languages such as Java, where any variable is cast as 
a given type and typing cannot change.

\subsection{Judgements}

All constructions in type theory are executed as judgements - rules that, given their preconditions are met, 
allow us to say something about the theory. The judgements are considered to be "external" statements in a 
meta-theory and as such they cannot interact with the theory itself. One way to think about this is to think 
of the theory as a game and the 
judgements as valid rules; you can get to many constructions starting from nothing, but are only ever allowed 
to make legal moves. 

There are four types of basic judgements in MLTT 
\begin{itemize}
    \item Type judgements, asserting that A is a type $\vdash A \; Type$
    \item (Definitional) equality of types, asserting that A and B are equal types $\vdash A \equiv B$
    \item Typing judgement, asserting that a is a term of type A $\vdash a : A$
    \item (Definitional) equality of terms, asserting that a and b are equal terms of type A $\vdash (a \equiv b) : A$ 
\end{itemize}

Sometimes two further judgements are added:
\begin{itemize}
    \item Context judgement, asserting that $\Gamma$ is a context $\vdash \Gamma \; ctx$
    \item Substitution judgement, asserting that $\theta$ is a substitution from $\Gamma$ to $\Delta$ $\Gamma \vdash \theta : \Delta $
\end{itemize}

Since judgements are handled as simple statements with pre-conditions and deterministic effects, the 
judgemental derivations are decidable. In particular, judgemental equality and type derivations are decidable.

All judgements in type theory are made in a context, usually denoted by capital greek letters ($\Gamma, \Delta$).
A context is a collection of previously made valid judgements and represents the preconditions of a given judgement.

We write $\Gamma \vdash J$ for judgement J in a context $\Gamma$
\subsection{Types}

To define a type one needs to give 4 judgements for that type: 
\begin{itemize}
    \item Type former - a way to construct a Type
    \item Token constructors - a way to construct canonical elements of a type 
    \item Token eliminator - a way to "use up" the canonical elements of a type 
    \item Computation rule - a way in which the prior two judgements interact 
\end{itemize}

Eliminators and computation rules are often bundled together into \textbf{recursors}(often called 
non-dependent eliminator) and \textbf{inductors}(often called dependent eliminator), which are procedures 
(functions) that give a generic way of constructing any function out of the type. 
 
\subsubsection{Basic Types}

Falsity, or bottom $\bot$, is a type with no introduction rules which by extension has no canonical elements.
A recursion principle for this type is that there are no functions with bottom as the domain; there are 0 
cases to consider as there are no canonical elements of it.

Falsity corresponds to a prototypical false proposition, an uninhabited type.

Unit type $\mathbbm{1}$ is a type with precisely one introduction rule - $*$ is an element of $\mathbbm{1}$. 
Consequently, the recursion principle for $\mathbbm{1}$ is that to make a function with $\mathbbm{1}$ as a 
domain you only need to provide an element to which $*$ is mapped.

Unit type corresponds to a prototypical true proposition, with the only canonical proof of it being $*$.

Universes are "large" types of "small" types which are an analogue of Grothendieck universes for type theory. One usually 
assumes an infinite hierarchy of universes $U_0 : U_1, U_1 : U_2, U_2 : U_3 ...$ with every higher universe 
type containing all the types in previous universes.

With Universe types type families can be thought of as functions $A(_) : X \rightarrow U$ and universes make 
some mathematical constructions such as a type of Groups possible. 

\subsubsection{Basic Type Constructors}
The type constructors of type theory are ways of constructing new types from old ones. Under the lense of 
propositions-as-types, those correspond to the binary connectives.

\subsubsection*{Product Type}

Conjunction is represented as a Product Type ($A \times B$). A product type consists of ordered pairs of 
elements, one from A and another from B. As such, a single type constructor for $A \times B$ is a pair

$a: A, b: B \vdash (a,b): A \times B$

An elimination rule for the product type is that given a canonical element $(a,b)$ of $A \times B$,
we can either return a or return b; these correspond to the two projector functions:

$p_1((a,b)) = a$

$p_2((a,b)) = b$

A recursion principle for product types states that a generic function from $A \times B$ may use both arguments,
and so recursor transforms a function with a type $A \rightarrow B \rightarrow C$ into a function with a type 
$(A \times B) \rightarrow C$. This is the un-currying operator.

\subsubsection*{Coproduct Type}
Disjunction is represented as Coproduct Type ($A + B$), otherwise known as disjoint union. Coproduct type 
consists of all the elements of A and B and contains an indicator from which disjunct each element is from. 
It is usually represented as ordered pairs, with first element either 0 or 1 depending on a type of a second 
element. There are two type constructors for coproducts: 
\begin{itemize}
    \item inl takes an element a of A and makes it into an element (0, a) of ($A + B$)
    \item inr takes an element b of B and makes it into an element (1, b) of ($A + B$)
\end{itemize} 

A recursor for Coproduct Type states that to construct a function from ($A + B$) into C, one needs to know 
what to do with elements of A and what to do with elements of B. As such, recursor takes in two functions of 
types $A \rightarrow C$ and $B \rightarrow C$ and returns a function of type $(A + B) \rightarrow C$.

\subsubsection*{Function Type}
Function types in MLTT are similar to ones in $\lambda$-calculus. There is 
one way to define a canonical element of a function type - $\lambda$-abstraction which, given an expression 
$\sigma$ of type A with a free variable x of type B (which also could be thought of as a family of terms of A 
indexed over elements of B) gives a function $\lambda x. \sigma$. The application of $\lambda$-abstraction type constructor 
binds the free variable from the expression it is applied to, eliminating it from the context.

An elimination rule for function types is function application - given a function $\phi: A \rightarrow B$ and 
a term $a: A$, we can apply $\phi$ to get $\phi (a) : B$.

A computation rule says that $\lambda x.\sigma (a) = \sigma [a/x]$.

Function type has a uniqueness principle - $\eta$-equality $f \coloneqq \lambda x.f(x))$. This judgemental 
equality cannot be derived from the rest of type theory,just like in $\lambda$-calculus and so is an optional 
addition. It does, however, make life a lot easier and so we adopt it in this paper. In the context of the 
univalence axiom, $\eta$-equality happens to be equivalent to function extensionality.

Note: Negation in constructive systems is represented as proposition implying absurdity, as opposed to 
having it's own unary logical connective. Hence, the type of "not A" is represented as $A \rightarrow \bot$.

\subsubsection{Dependant types}

As mentioned at the end of Section 2, in order for our Type Theory to be able to handle predicates, we need to 
include dependant types. In MLTT those are dependant products and dependant functions, corresponding to 
quantifiers in logic.

\subsubsection*{Dependant Function Type}

Dependant functions are a more general type of functions where the type of codomain varies with the element of 
the domain. Given a type A and a family of types $B: A \rightarrow \mathcal{U}$, we can construct a type of 
dependant functions $\Pi_{x_A}B(x)$. As per BHK interpretation, such dependant functions correspond to universal 
quantification. 

The constructor and eliminator for this type are the same as in the non-dependant function 
case, except the expression to which  $\lambda$-abstraction is applied needs not have a consistent type.

Same thing goes for eliminator being function application and the computation rule.

One important feature of the dependant functions is that they allow us to define polymorphic functions. In 
order to do so, the dependent function might take a Type (as an element of the universe) as the first input and 
produce a function with that type as a domain. An example of such a polymorphic function is a generic identity:

$\lambda(A:\mathcal{U}).\lambda(x:A).x$

which takes a type as an argument and returns an identity function for that type.

\subsubsection*{Dependant Product Type}

As per BHK interpretation, representing existential quantifiers are the Dependant Products, the types of pairs 
where a second element of a pair can have a varying type depending on the first element $\Sigma_{x:A}B(x)$. 
Similar to the regular products, dependant products 
have a single element constructor:

$a: A, b: B(a) \vdash (a,b): \Sigma_{x:A}B(x)$
s
A recursor for dependant product, is the un-currying operator, for dependent functions. 

A formal definition of this recursor goes like so:

A type of the recursor is 

\[ rec_{\Sigma_{x:A}B(x)}(C, g, (a,b)): \Pi_{C:U} (\Pi_{x:A}B(x) \rightarrow C) \rightarrow (\Sigma_(x:A)B(x)) \rightarrow C \]

Where C is the target type, g is a curried function and (a, b) is the element to which the resulting function 
is applied.

The defining equation for this is 
\[ rec_{\Sigma_{x:A}B(x)}(C, g, (a,b)) \coloneqq g(a)(b)\] 

\subsection{Identity type}

Identity type is the defining feature of intuitionistic type theory; as such, it warrants 
it's own chapter.

It may seem that constructions above comprehensively cover all the logical connectors of first-order logic, 
and as such are sufficient for representing the predicate logic in it's entirety. What this assessment, however, 
misses is that first order logic has another often overlooked primitive symbol - propositional equality.
None of the constructions above are able to handle general propositional equality. Judgemental 
equality equating terms with same canonical representations is too weak of a notion, at least because it must by definition
be decidable. It is also only applicable to closed terms in an empty context - there is no way, for example, to show 
n + 2 judgmentally equals 2 + n. To handle those more complex equalities between terms we must introduce a new predicate -
an identity type. We want a type family that, given any 2 elements a and b of type A represents a type of proofs that
a = b. We call this type Eq(A, a, b) or alternatively $a =_A b$ or $Id_A(a, b)$. It makes sense to define equality to be the smallest 
reflexive, symmetric and transitive relation on a type. As such, the only type constructor for the identity type family 
is reflexivity:
\begin{prooftree}
    \AxiomC{refl(a) : $a =_A a$}
    
\end{prooftree}
Symmetry and transitivity of equality can then be proven as theorems - we can construct functions of types that we will 
demonstrate later.
\begin{align*}
    (x=_Ay) \rightarrow (y=_Az) \rightarrow (x=_Az)\\
    (x=_Ay) \rightarrow (y=_xz)
\end{align*}

\subsubsection{Homotopy Interpretation}

One distinctive feature of MLTT is that unlike in classical mathematical theories, propositional equalities in 
MLTT possess computational content; they are not merely an equivalence relation on terms, but a type with it's 
terms possessing computational meaning defined in the type's recursor. This means that in general the type of 
equalities has many inhabitants, a notion wholly dissimilar to the classical conception of equality.

To give a semantic explanation of this, one needs to think of elements of an identity types as ways of 
identifying the two elements rather than mere propositions of their equality. 

To demonstrate this, let's consider a type of $a+(b+c) =_{\mathbb{N}} (a+b)+c$, representing the 
associativity of Natural numbers. One can obtain an inhabitant of such a type by unwrapping the inductive 
definition of addition and doing an induction on c. Alternatively, one can use commutativity to demonstrate 
that a+(b+c) = (c+b)+a and (a+b)+c = c+(b+a) and do induction on a instead. These two proofs cannot be 
judgmentally equated and as such they are different elements of a type with different computational content. 
These two proofs, however, are in their essence "the same proof" except one goes about it in a more 
roundabout way. This notion can be formalised in an equality type between these two proofs. In fact, in general 
case we may have an infinite hierarchy of equality types, each layer proving the equality of terms in the 
previous layer.

Moreover, since propositional equality is a type, it makes sense to define identity types on identity types, 
creating an infinite hierarchy of identifications of different levels. This hierarchy of equalities is usually 
interpreted as a homotopy type (alternatively, an $\infty$-groupoid), with types corresponding to homotopy 
spaces, element of a type to points in spaces, propositional equalities as paths between points and higher 
order equalities as homotopies. This interpretation is especially fitting since, as we will see in the next 
subsection, one cannot in general propositionally distinguish between equivalent elements and hence in 
particular cannot distinguish between two types with the same homotopy structure.

This interpretation gives rise to an extensive cross-pollination between the subjects of type theory and 
homotopy theory, with type theory being a more natural language to express many homotopical theorems than 
set theory is and with homotopy theory suggesting many extensions that can be applied to type theory, such as 
a univalence axiom. 

\subsubsection{Path Induction}
On one hand it makes sense that we should only equate terms which are the same, but on the other hand it is unclear how
this is not a trivial relation and how it is different from judgemental equality. What must be noted, however, is that 
we do not define these propositional equalities in isolation; we define the entire family of identity types dependant on 
elements of a given type at the same time. This allows the identity type to pick up the structure in the underlying 
type and lifting it to the level of equality. 

A recursion principle for the identity type is called the J-rule or Path Induction and it states that in order to prove something 
for all elements of $Id_A(a, b)$ where a and b may vary it is enough to prove it for $refl(a, a)$ for all a.

\subsubsection{Based Path Induction}

An alternative and equivalent definition for identity induction principle is Based Path Induction, which is the 
same as regular Path Induction, except one end of a path is fixed. We will primarily be using this form of 
path induction, as it requires recursion through one less variable. The proof for equivalence of the two can 
be seen in HoTT book \cite{UFP13}.

Formally, given a family of based paths  
\[E_a = \Sigma_{x: A}(x =_A a)\] 
, in order to define a function on this type family one needs only define this function on $(a, refl_a)$. So, 
given an element 
\[ c: C(a, refl_a)\]
We can obtain a function 
\[ f : \Pi_{(x:A)}\Pi_{(p: a=_A x)}C(x, p)\]
s.t. it produces the desired  result of the base case:
\[ f(a, refl_a) \coloneqq c\]

This induction principle, when packaged as an inductor function looks like so:

\[ind_{=_A} : \Pi_{(a: A)}\Pi_{C: \Pi_{(x:A)}(a =_A x) \rightarrow \mathcal{U}} C(a, refl_a) \rightarrow \Pi_{(x:A)}\Pi_{(p: a=_A x)}C(x, p)  \]

with defining equality 

\[ ind_{=_A}(a, C, c, a, refl_a) = c\]

\subsubsection{Properties of Identity}

Equiped with path induction, we are now prepared to prove the fundamental properties of the Identity type.

\begin{lemma}
    Propositional equality is an equivalence relation. 
\end{lemma}

\begin{proof}
    Reflexivity is immediate by the constructor of the identity type.

    For symmetry, we need to construct a function of the type $sym: (x = y) \rightarrow (y = x)$ for every type A
    and every tokens x and y. We do this by pattern matching - to construct such a function, by path induction,
    we only need to consider one special case of $refl_x: (x = y)$ which we send to $refl_y: (y=x)$, producing 
    a function of the needed type.

    Transitivity is done in the same way, to construct a function $tran: (x = y) \rightarrow (y = z) \rightarrow (x = z)$,
    we apply path induction twice. We need to construct such function for all $x, y, z: A$ and all $p : x = y $
    and $q : y = z$. By path induction on p we may assume y is x and p is $refl_x$; then, by induction on q 
    we may assume z is x and q is $refl_x$. Then, we may define a function to be $tran(refl_x, refl_x) = refl_x$, 
    finishing the proof.
\end{proof}

The key to understanding the Identity type and the role it plays in HoTT is to see how it interacts with other 
type theoretic constructions. Informally, Identity types represent the indistinguishables of Type Theory; all 
constructs possible in the language of type theory respect the indistinguishability of elements for which the 
identity type is inhabited in the sense that it is impossible to prove a proposition involving one of the 
identical elements but not the other. We will later make this notion precise using transports and Univalence 
axiom. 

\subsubsection*{Uniqueness principle}

Every type in MLTT has a uniqueness principle, stating that every element of that type is equivalent to some 
canonical element of that type. An example of a uniqueness principle for a coproduct type is:

\begin{lemma}
    For every element $c : A + B$, there exists either an element of type $c =_{A+B} a$ for some $a : A$ or an 
    element of type $c =_{A+B} b$ for some $b: B$.
    
    \[ \Pi_{c : A+B}( (\Sigma_{a:A} c =_{A + B}inl(a)) + (\Sigma_{b:B} c =_{A+B}inr(b ))   ) \]
\end{lemma}

\begin{proof}
    We use the recursion principle for coproduct types by defining 
    \[uniq_{A+B}(inl(a)) = refl_inl(a)\] 
    and 
    \[uniq_{A+B}(inr(b)) = refl_inl(b)\]
    Then, for every element c of $A + B$ we have 
    \[uniq_{A+B}(c) : (\Sigma_{a:A} c =_{A + B}inl(a)) + (\Sigma_{b:B} c =_{A+B}inr(b ))\]
\end{proof}

This principle, when viewed through a lens of homotopy interpretation means that every type has a "skeleton" 
generated by it's constructors, and the entire type is homotopically equivalent to this skeleton; furthermore, 
it states that it makes sense to in general work "up to homotopy".

Similarly, we have a uniqueness principle for identities which is a little tricky to state. In order to state 
uniqueness principle for identities we want to show that every element of $a =_A b$ is of the form 
$refl_a : a =_A a$, but the two above types are different. To be able to state the uniqueness principle we need 
to be working in a larger type $E_a \coloneqq \Sigma_{x:A}(a =_A x)$ of equalities with a. Then we have:

\begin{lemma}
    Let $a:A$ and $E_a \coloneqq \Sigma_{x:A}(a =_A x)$. Then there is an element 
    \[ uniq : \Pi_{(x,p): E_a} (x, p) =_{E_a} (a. refl_a)\]
\end{lemma}

\begin{proof}
    A proof of this lemma is outside the scope of this dissertation as it requires a bit of Category Theory. 
    The proof can be found as Lemma 2.3.2 in \cite{UFP13}.
\end{proof}

\subsubsection*{Functoriality and substitution of equals for equals}

Here we make the notion of "indistinguishabiltiy of identicals" or "Leibniz principle" more precise; in 
particular, we will show that functions respect propositional equality and that predicates cannot differentiate 
between propositionally equal elements.

\begin{lemma}[Functoriality of identity \cite{UFP13}]
    Let $f: A \rightarrow B$ be a function. Then there is a family of functions
    \[ ap_f : \Pi_{x,y:A}(x =_A y) \rightarrow (f(x) =_B f(y))\]
    s.t. $ap_f(refl_x) = refl_{f(x)}$
    That is, for all the identical elements x, y their images under f are also identical.
\end{lemma}

\begin{proof}
    Let $P: \Pi_{x,y:A}(x =_A y) \rightarrow U$ be given by $P(x, y, p) \coloneqq (f(x) =_B f(y))$.

    We can construct an element of $\Pi_{x:A}P(x,x,refl_x)$ as $\lambda x.refl_f(x)$. Then by path induction, 
    exists \[ ap_f : \Pi_{x,y:A}(x =_A y) \rightarrow (f(x) =_B f(y))\] as required and $ap_f(refl_x) = refl_{f(x)}$
    by the computation rule.
\end{proof}

Viewed frim the lens of Homotopy interpretation, we can read this result as all definable functions being 
continuous.

\begin{lemma}[Transports \cite{UFP13}]
    Given a type family $P : A \rightarrow U$ (which can be viewed as a predicate in A), whenever $x, y : A$
    are identical, that is there is an element $p: x =_A y$ there exists a function 
    \[ p_* : P(x) \rightarrow P(y)\]
\end{lemma}

We can read this lemma as "whenever two elements are related by an identity type, they are propositionally 
indistinguishable in a sense that if a predicate is true for one of them, it is also true for another".

\begin{proof}
    By path induction it suffices to show the result for $x \coloneqq y$ and $p \coloneqq refl_x$. Then we can take 
    \[ p_* = \lambda (a: P(x)).a : P(x) \rightarrow P(x)\]
\end{proof}

The above results show that Identity Type as defined by us captures the indiscernibility in the sense that the 
language of MLTT treats elements related by identity type as indiscernible. In fact, the definition of Path 
Induction is designed precisely in a way to capture this:  

\begin{theorem}
    In a language of MLTT without path induction, path induction is equivalent to indiscernibility of identicals 
    (existence of transports) and uniqueness principle for identity types
\end{theorem}

\begin{proof}
    We showed that path induction implies the other two in the lemmas above.

    Now, let us assume the existence of transports and the uniqueness principle for identity types.
    As before, let $E_a = \Sigma_{x:A}(a =_A x)$. Let P be any predicate on $E_a$ and $(b,p)$ be any element 
    of $E_a$.

    By existence of transports, there is a map 
    \[ f: (((b, p) =_A (a, refl_a))) \rightarrow P((a, refl_a)) \rightarrow P((b, p))\]

    By the uniqueness principle we get that $(a, refl_a) =_A (b, p)$ is inhabited. By symmetry of identity, so 
    is $(b, p) =_A (a, refl_a)$. Let the inhabitant of that type be q. Then 
    \[ f(q) : P((a, refl_a)) \rightarrow P((b, p))\] 
    as desired.

    Thus, if we want to show that a predicate P holds for every $(b, p) : E_a$, we only need to show that P
    holds for $(a, refl_a)$ giving us the statement of based path induction.
\end{proof}

\subsection{Inductive types}

In order to define more complicated types such as Natural Numbers, one can employ the inductive definitions. A
generic inductive type is "freely generated" by it's generators, some of which can take elements of that same 
type as inputs. A generic element of such a type, then, is a tree-like structure of generators, with leafs 
representing the constructors that do not take elements of it's own type as inputs and all the other nodes 
representing ones that do. 

The easiest example of such an inductive type is Natural Numbers ($\mathbb{N}$), generated by 0 (a constant constructor) and 
succ (a function from $\mathbb{N}$ to $\mathbb{N}$). A generic member of $\mathbb{N}$ is then a tree consisting 
of any number of succ and 0 nodes; in this particular case, since the arity of succ is 1, all such trees are 
non-branching.

The recursive principle for a generic inductive type defines a function on that type by specifying what a 
function must do on the generators; this suffices since the generators of an inductive type act as a "spanning
set", akin to a basis in a vector space in that every (canonical) elemnt of the type can be written as a 
successive application of generators.

Such recursors act similar to fold functions from functional programming, taking in a tree-like element of an 
inductive type and replacing every generator in it's construction by a function of a required arity. A recursor 
for a list, for example, in order tp generate a function $List[A] -> B$ replaces every concatenation generator 
with a function $g :B -> A -> B$ and the empty list constructor with an element of B:

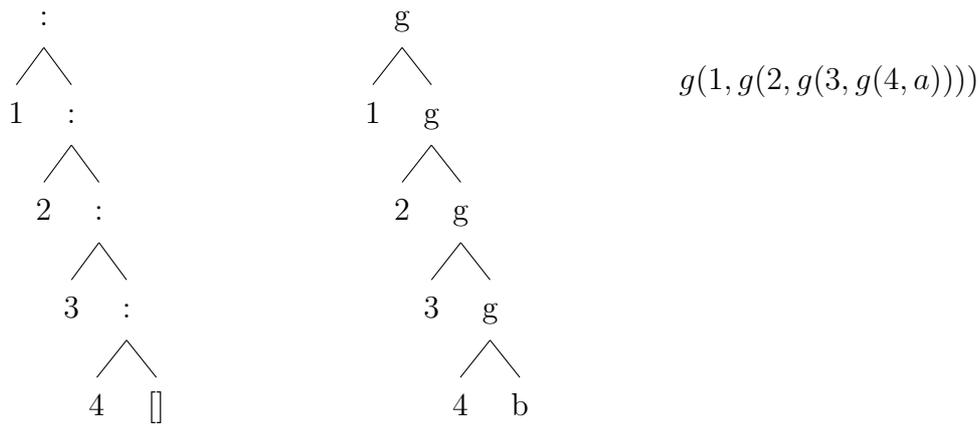
\begin{multicols}{3}
    \begin{forest}
        [:
            [1]
            [:
                [2]
                [:
                    [3]
                    [:
                        [4]
                        [$\lbrack \rbrack$]
                    ]
                ]
            ]
        ]
    \end{forest} \break

    \begin{forest}
        [g
            [1]
            [g
                [2]
                [g
                    [3]
                    [g
                        [4]
                        [b]
                    ]
                ]
            ]  
        ]
    \end{forest}

    \[ g(1, g(2, g(3, g(4, a)))) \]
\end{multicols}

All the type constructors mentioned above (except function types and identity types) are a specific example 
of a general recursive type where all the constructors are arity 0.

\subsubsection{Uniqueness of inductive types}

Above we mentioned that natural numbers $\mathbb{N}$ can be defined as an inductive type freely generated by 
$0 : \mathbb{N}$ and $succ : \mathbb{N} \rightarrow \mathbb{N}$. However, there is nothing preventing us from 
defining a different datatype $\mathbb{N}'$ freely generated by $0' : \mathbb{N'}$ and 
$'succ' : \mathbb{N}' \rightarrow \mathbb{N}'$. Such datatype will have an identical looking recursion principle 
and will satisfy all the properties $\mathbb{N}$ does, but would be syntactically distinct. 

This highlights the fact that the identity type as defined by Martin L\"{o}f is too restrictive. It identifies 
functions based on their definition, while mathematicians usually identify functions extensionally; it also 
identifies types based on their presentation, while mathematicians usually identify sets based on isomorphism.

Such isomorphic constructions arise in the mathematics all the time, be it natural numbers and non-negative 
integers, lists and vectors of all sizes, etc. We, as mathematicians, know that those are really "the same" 
and would use the two interchangeably. This is broadly unproblematic, but it does sweep a lot of latent 
computational content under the rug.

When two types are isomorphic in a way shown above, we can identify them by defining two functions 
\begin{itemize}
    \item $f \coloneqq rec_{\mathbb{N}}(\mathbb{N}', 0', succ')$
    \item \item $f \coloneqq rec_{\mathbb{N'}}(\mathbb{N}, 0, succ)$
\end{itemize}
that coerce between the two types by replacing the constructors of one with constructors of another. 

With that, we can freely substitute one type for another by transferring every function defined on one type 
to another, applying the function, and then bringing the result back:

\[ double' \coloneqq \lambda n.f(double(g(n)))\]

Explicit handling of such isomorphism is quite cumbersome, as one needs to always coerce betwen the two 
isomorphic types explicitly. This motivates many extensions of type theory with coarser notions of equality, 
such as \cite{altenkirch1999extensional} and \cite{hofmann1995extensional}. Tthe most widespread extension of
this type is borrowed from Homotopy Theory by Voevodsky and is a Univalence Axiom presented below.

\subsection{Univalence Axiom \cite{Martin18}}

In order to formally state this axiom, we need first talk about equivalences.

\begin{definition}[Homotopy \cite{UFP13}]
    Let $f,g : \Pi_{x:A}P(x)$ be two (dependent) functions. A homotopy from f to g is a dependant function of 
    type 
    \[ (f \sim g) \coloneqq \Pi_{x: A} (f(x) = g(x))\]
\end{definition}
A homotpy captures the idea of two functions being extensionally equivalent up-to homotopy; this notion \
coincides with two functions being homotopical in homotopy theory.

\begin{definition}[Equivalence \cite{UFP13}]
    Let f be a function $f: A \rightarrow B$. f is an equivalence if the following type is inhabited:
    \[ isequiv(f) \coloneqq (\Sigma_{g: B \rightarrow A} (f \circ g \sim id_B)) \times (\Sigma_{h: B \rightarrow A} (h \circ f \sim id_A))\]
    
\end{definition}

Equivalence, in it's essence, captures an idea of an isomorphism between two types, but "up to homotopy"; it 
ensures the existence of inverses up-to-homotopy. In fact, left inverses will always be right inverses and vice 
versa; the two are not required to be the same in the above definition due to problems with higher coherence 
in proofs involving equivalences.

\begin{definition}
    The two types A and B are said to be equivalent if there is an equivalence between them; that is, if the 
    below type in inhabited:

    \[ A \cong B \coloneqq \Sigma_{f: A \rightarrow B}isequiv(f)\]

\end{definition}

\begin{lemma}
    The following type is inhabited:
    \[ IdToEq : \Pi_{X, Y: U} (X =_U Y) \rightarrow (X \cong B)\]
\end{lemma}
\begin{proof}
    Consider an inhabitant p of $X =_U Y$. Let $p_*$ be the associated transport function. Then $p_*$ is an 
    equivalence.

    To prove this, we only need to consider a special case when $p = refl_x$, in which case $p_* = id_X$ which 
    is an equivalence with it being it's own left and right inverse.
\end{proof}

\begin{definition}[Univalence]
    Univalence is a property of the universe type. It is represented by a type 
    \[ isUnivalent(U) \coloneqq \Pi_{X, Y : U}isEquiv(IdToEq(X, Y))\]
\end{definition}

In general, universe may or may not be univalent. Univalence of universes, however, is consistent with MLTT and 
so it is impossible to disproof. As such, the type coercion captured by univalence always exists in the theory, 
but univalence makes it more accessible.

Univalence axiom asserts that the universe is univalent. 

\begin{definition}[Univalence Axiom]
    There is an element 
    \[ univ : isUnivalent(U)\]    
\end{definition}

\subsection{Sets in Type Theory}

Since in general type theory has an infinite-dimensional path structure of identity types, it is often easier 
to work with truncations of it.

\begin{definition}[Mere Proposition]
    In type theory we call a type \textbf{mere proposition} if it has only one element up to equality. Such 
    types satisfy a predicate
    \[ isProp(P) \coloneqq \Pi_{x, y: P} x =_P y\]
\end{definition}

A type is mere proposition if it is equivalent to either unit type (truth) $\mathbbm{1}$ or bottom type 
(falsehood) $\bot$.

Any proposition in type theory can be reduced to a mere proposition by propositional truncation:

\begin{definition}[Propositional Truncation]
    Propositional truncation of a type A is a type $\|A\|$; we have that for every $a:A$ there is $|a| :\|A\|$
    and for every $x, y : \|A\|, x =_{\|A\|} y$
\end{definition}

Propositional truncation makes every type function like a proposition in classical logic. In particular, many 
results of classical logic, such as LEM, are true for mere propositions. Such a transformation, however, erases 
the computational content of the proposition and as such breaks constructivity of terms that use such a type.

\begin{definition}[Set]
    A set is a type for which every identity type is a mere proposition:
    \[ isSet(A) \coloneqq \Pi_{x, y: A}isProp(x =_A y)\]
\end{definition}

Sets in type theory behave similar to normal sets, with no computational content contained in equality. 

In general, instead of cutting off identities at the first level, one may instead truncate at n'th level, 
ontaining an n-truncated type.

\begin{definition}[n-truncated type]
    A type is 1-truncated if it is a set. A type is n+1-truncated if for all elements of that type the identity 
    type is n-truncated.
\end{definition}

\section{Canonicity and Computability of Type Theory}

In this section we discuss the univalence axiom, what it means for HoTT and it's inherent problems. 

\subsection{Canonicity and Computability}
Dependant Type Theory in it's most basic form enjoys some great computational properties. It has decidable 
judgemental equality and type checking (which is required by Martin L\"{o}f's meaning explanation) and it 
furthermore has a property of canonicity.

\begin{definition}[Canonical Term]
    Canonical terms are of a type are exactly those introduced by a contructor. There are two ways of defining 
    canonical terms : 
    
    \begin{itemize}
        \item Canonical terms are closed terms produced entirely by constructors  (such as S(S(0))). This 
definition corresponds to the concept of eager evaluation - for a term to be considered canonical it needs to 
be "fully evaluated"
        \item Canonical terms are closed terms introduced by a constructor, no matter if the input of the 
constructor is fully evaluated. For example, $\lambda x. x + s(0))$ is considered a canonical term of 
$\mathbb{N} \rightarrow \mathbb{N}$. This definition corresponds to the concept of lazy evaluation, since the 
input of a function needs not be always fully evaluated.
    \end{itemize}
    These two ways of defining closed terms are equivalent in base MLTT and only correspond to different 
    evaluation strategies. As such, Lazy evaluation is usually adopted as default. They are, however, different 
    in HoTT since the theory looses canonicity.
\end{definition}

\begin{definition}[Canonicity]
    A Type Theory is said to enjoy canonicity if every closed term of a given type reduces to a canonical term.
\end{definition}

Canonicity corresponds to an idea that every computation in type theory terminates; every program has a 
canonical term to which it reduces which indicates the result of computation. For example, $2 + 2$ reduces to 
$4$.

In this context, judgemental equality can be seen as precisely equating all the terms with same canonical form 
(as one would expect).

Canonicity of MLTT is a product of all constructions in it being defined in a style of natural deduction - 
with only way of introducing new terms being the canonical terms of a type, which always have a corresponding 
eliminators and computation rules. This is precisely the reason why one in general avoids adding axioms to a 
type theory.

\subsection{Decidability}

As mentioned above, judgemental equality and type checking in MLTT is decidable. This is, similar to canonicity,
thanks to a strict way in which the theory is defined as a natural deduction style system. Every term, by very  
nature of its construction, possesses a type since only way of introducing new terms is using strictly typed 
cosntructors. Decidability of judgemental equality follows from canonicity, but it can be guaranteed in 
different ways, similar to type checking.

\subsection{Problem with Univalence}

As we talked about above, Univalence makes life of mathematicians a lot easier, allowing them to convert 
between equivalent structures freely, while leaving all the actual computational coercion implicit and letting 
the language figure it out.

When considering type theory as a language, univalence is also incredibly useful as it allows the developer to 
reuse any code written for one type (class) to be used for every isomorphic type. Normally, one would require 
to explicitly coerce between the two types and prove for each construction that it respects this coercion.

A problem with axiomatic handling of univalence, however, is that it presents a term of a univalence type which claims 
to do the computational lifting described above, but in fact contains no computaional content; as such, every 
program that uses univalence to construct a proof of identity of two isomorphic structures will inevitably get 
stuck trying to do the substitution of one structure for another, as there is no eliminator for identity proof 
produced by univalence. Univalence is a well-typed but computationally stuck term, unable to reduce to a 
Canonical form, undermining the canonicity of the theory.

A judgemental handling of univalence is currently one of the biggest research areas in HoTT, as it is the 
biggest cornerstone towards the full power of HoTT being utilised in a computational context, such as proof 
assistance. There are a few modifications to HoTT that obtain a judgemental notion of univalence such as 
cubical type theory, but no judgemental presentation compatible with Martin L\"{o}f's type theory has been 
found yet.

\subsection{Two-Dimensional Type Theory}

Two-Dimensional (groupoid-level) type theory (2TT) \cite{Licata12} aims to give a judgemental presentation of 
equivalence for a 2-dimensional 
type theory, making univalence a provable result rather than an axiom. 

To do this, equivalence is considered judgmentally, with a judgement of a form 
$\Gamma \vdash \alpha : M \simeq_A N$
meaning that $\alpha$ is the evidence of equivalence of M and N as objects of A is added. The language is then 
extended with judgements ensuring that type and term families respect equivalences. Identity type is then 
interpreted as a hom-type of equivalences.

Key innovation of 2TT is a new judgement of a form $\Gamma \vdash \alpha : M \simeq_A N$ which states that 
$\alpha$ is the evidence of equivalence of M and N as objects of A. A computational content of such equivalence 
is then made explicit by providing judgemental operations ensuring that every type and term family is functorial
in their indices.
Univalence and function extensionality then follow as rules of equivalence and their actions account for a 
computational content in equivalence proofs.

There are 3 ways of introducing equivalences, corresponding to rules of an equivalence relation :
\begin{itemize}
    \item \[ \vdash refl_{\theta}^{\Delta} : \theta \simeq_{\Delta} \theta\]
    \item \[ \delta : \theta_1 \simeq_{\Delta} \theta_2 \vdash \delta^{-1} : \theta_2 \simeq_{\Delta} \theta_1 \]
    \item \[ \delta_1 : \theta_1 \simeq_{\Delta} \theta_2, \delta_2 : \theta_2 \simeq_{\Delta} \theta_3 \vdash \delta_2 \circ \delta_1 : \theta_1 \simeq_{\Delta} \theta_3\]
\end{itemize}
Furthermore, these constructors respect associativity and inverse laws, making equivalence structure behave like 
a groupoid.

Judgemental preservation of equivalence requires that type and term families are functorial in their indices. 
This is guaranteed by the equivalence eliminators defined as operations  map for type families and resp for 
term families:

\begin{prooftree}
    \AxiomC{$\Gamma \vdash B(x) : A \rightarrow U$}
    \AxiomC{$\Gamma \vdash \alpha: M_1 \simeq_A M_2$}
    \AxiomC{$\Gamma \vdash M : B[M_1/x]$}
    \TrinaryInfC{$\Gamma \vdash map_{x: A.B} \alpha M : B[M_2/x]$}
\end{prooftree}

\begin{prooftree}
    \AxiomC{$\Gamma, x:A \vdash F : B$}
    \AxiomC{$\Gamma \vdash \alpha: M_1 \simeq_A M_2$}
    \BinaryInfC{$\Gamma \vdash resp (x.F) \alpha : F[M_1/x] \simeq_B F[M_2/x]$}
\end{prooftree}

\begin{remark}
    In order to define above maps more easily, 2DTT is presented as a substitution calculus, with a judgement 
    for context substitution.
\end{remark}

Each type constructor is then equipped with judgements ensuring that it behaves properly with equivalences and 
above maps. One can find the entire formulation of this language in the original paper \cite[(Pages 3-5, 11)]{Licata12}, 
together with 
the explanation of all the constructions.

Judgements of 2TT ends up having a simple categorical interpretation in a category of groupoids \cite{licata20112}:
\begin{itemize}
    \item Context judgement $[\![ \Gamma]\!]$ is a category.
    \item Substitution judgement $[\![\Gamma \vdash \theta : \Delta]\!]$ is a functor $[\![\theta]\!] : [\![\Gamma]\!] \rightarrow [\![\Delta]\!]$
    \item Equivalence judgement $[\![\Gamma \vdash \delta : \theta_1 \simeq_{\Delta} \theta_2]\!]$ is a natural transformation $[\![\delta]\!] : [\![\theta_1]\!] \simeq [\![\theta_2]\!]$
    \item Type judgement $[\![A]\!]$ is a functor $[\![A]\!] : [\![\Gamma]\!] \rightarrow GPD$ where GPD is a large category of groupoids and fucntors
\end{itemize}

I omit full details of a proof, but it can be shown that a Boolean type gets interpreted as a discrete category 
with two objects, thus proving consistency (statement "true = false" is refutable in one semantic interpretation,
and as such not provable).

\subsubsection{Canoncirity}

In order to prove canonicity, we defines a semantic interpretation of 2DTT into syntactically presented 
groupoids and functors, saying that a logical expression is reducible if it is a member of these semantic domains;
open terms are reducible if they take reducible terms to reducible terms.

\subsubsection*{Syntactically Presented Groupoids and Functors}

\begin{definition}
    A Groupoid $\langle \Gamma \rangle$ is presented by context $\Gamma$ iff:
    \begin{itemize}
        \item $Ob(G)$ is a subset of the equivalence classes of substitutions $\bullet \vdash \theta : \Gamma$ modulo definitional equality (such an equivalence class is denoted $\llangle \theta \rrangle$ from now on). 
        \item Set of morphisms $\llangle \theta_1 \rrangle \rightarrow_G \llangle \theta_2 \rrangle$ is a subset of equivalences classes of equivalences $\bullet \vdash \theta_1 \simeq_{\Gamma} \theta_2$ modulo definitional equality.
        \item Identity at $\llangle \theta \rrangle$, composition and inverses are given by equivalence type constructors.
    \end{itemize}
\end{definition}

\begin{definition}
    A groupoid $\langle A \rangle$ is presented by type A iff:
    \begin{itemize}
        \item $Ob(G)$ is a subset of the equivalence classes of terms of A modulo definitional equality. 
        \item Set of morphisms $\llangle M_1 \rrangle \rightarrow_G \llangle M_2 \rrangle$ is a subset of e$\llangle M_1 \simeq_{\Gamma} M_2 \rrangle$.
        \item Identity at M, composition and inverses are given by equivalence type constructors.
    \end{itemize}
\end{definition}

\begin{definition}
    A functor $\langle M \rangle : \langle A \rangle  \rightarrow \langle B \rangle$ is presented by $x : A \vdash M : B$ iff 
    \begin{itemize}
        \item For all $\llangle N \rrangle \in Ob(\langle A \rangle), \\ \langle M \rangle(\llangle N \rrangle) = \llangle M[N / x] \rrangle$
        \item For all $\llangle \alpha \rrangle : \llangle N_1 \rrangle \rightarrow_{\langle A \rangle} \llangle N_2 \rrangle, \\ \langle M \rangle(\llangle \alpha \rrangle) = \llangle M[\alpha / x] \rrangle$ 
    \end{itemize}
\end{definition}

\begin{definition}
    A functor $\langle A \rangle : \langle \Gamma \rangle  \rightarrow GPD$ is presented by type A (where GPD stands for category of groupoids and functors) iff 
    \begin{itemize}
        \item For all $\llangle \theta_1 \rrangle \in Ob(\langle \Gamma \rangle), \\\langle A \rangle(\llangle \theta_1 \rrangle)$ is presented by $A[\theta_1]]$
        \item For all $\llangle \delta \rrangle : \llangle \theta_1 \rrangle \rightarrow_{\langle \Gamma \rangle} \llangle \theta_2 \rrangle, \\ \langle A \rangle(\llangle \delta \rrangle)$ is presented by $x : A[\theta_1].map_A \delta_x$ 
    \end{itemize}
\end{definition}

\begin{definition}
    Given groupoids $\langle \Gamma \rangle$ and $\langle \Delta \rangle$, the set of reducible substitutions 
    $RedSub(\langle \Gamma \rangle, \langle \Delta \rangle)$ is those substitutions $\Gamma \vdash \theta : \Delta$ for which 
    \begin{itemize}
        \item For all $\llangle \theta_1 \rrangle \in Ob(\langle \Gamma \rangle), \\ \llangle \theta[\theta_1] \rrangle \in Ob(\langle \Delta \rangle)$
        \item For all $\llangle \delta \rrangle : \llangle \theta_1 \rrangle \rightarrow_{\langle \Gamma \rangle} \llangle \theta_2 \rrangle, \\ \llangle \theta[\delta] \rrangle : \llangle \theta[\theta_1] \rrangle \rightarrow_{\langle \Delta \rangle} \llangle \theta[\theta_2] \rrangle$
    \end{itemize}
\end{definition}

\begin{definition}
    Given groupoids $\langle \Gamma \rangle$ and $\langle \Delta \rangle$ and $\theta_1, \theta_2 \in RedSubst \langle \Gamma \rangle, \langle \Delta \rangle)$
    define the set of reducible equivalences $RedEquiv_{\langle \Gamma \rangle}^{\langle \Delta \rangle}(\theta_1, \theta_2)$
    to contain those equivalences $\Gamma \vdash \delta : \theta_1 \simeq_{\Delta} \theta_2 $s.t. for all $
    \llangle \delta' \rrangle : \llangle \theta_1' \rrangle \rightarrow_{\langle \Gamma \rangle} \llangle \theta_2' \rrangle, \\ \llangle \delta[\delta'] \rrangle : \llangle \theta_1[\theta_1'] \rrangle \rightarrow_{\langle \Delta \rangle} \llangle \theta_2[\theta_2'] \rrangle$
\end{definition}

\begin{definition}
    Given a groupoid $\langle \Gamma \rangle$ and a functor $\langle A \rangle : \langle \Gamma \rangle \rightarrow GPD$,
    define a set of reducible terms $RedTm^{\langle \Gamma \rangle}$ to be the terms $\Gamma \vdash M :A $ s.t.
    \begin{itemize}
        \item For all $\llangle \theta_1 \rrangle \in Ob(\langle \Gamma \rangle), \llangle M[\theta_1] \rrangle \in Ob(\langle A \rangle (\llangle \theta_1 \rrangle))$
        \item For all $\llangle \delta \rrangle : \llangle \theta_1 \rrangle \rightarrow_{\langle \Gamma \rangle} \llangle \theta_2 \rrangle, \\
        \llangle M[\delta] \rrangle : \llangle map_A (\delta M [\theta_1] \rrangle \rightarrow_{\langle A \rangle (\llangle \theta_2 \rrangle)} \rrangle : \llangle M [\theta_2] \rrangle$
    \end{itemize}
\end{definition}

\begin{theorem}[Fundamental Theorem]
    There exist partial functions $[\![\Gamma ]\!]$  and $[\![ A]\!]$ which send contexts and terms respectively 
    to their presentations while respecting equivalences and making all well-typed expressions reducible:
    \begin{itemize}
        \item If $\Gamma$ is a context, then $[\![\Gamma ]\!]$ is a groupoid presented by $\Gamma$
        \item If $\Gamma \equiv \Gamma',  [\![\Gamma ]\!] = [\![\Gamma' ]\!]$ 
        \item If $\Gamma \vdash \theta : \Delta$, then $\theta \in RedSebst([\![\Gamma ]\!], [\![\Delta ]\!])$
        \item If $\Gamma \vdash : \theta_1 \simeq_{\Delta} \theta_2$, then $\delta \in RedEquiv_{[\![\Delta ]\!]}^{[\![\Gamma ]\!]}(\theta_1, \theta_2) $
        \item If $\Gamma \vdash A \: Type $, then $[\![A]\!]$ is a functor $[\![\Gamma ]\!] \rightarrow GPD$ presented by A
        \item If $\Gamma \vdash A \equiv A' \: Type$, then $[\![A]\!] = [\![A']\!]$
        \item If $\Gamma \vdash M : A $, then $M \in RedTm^{[\![\Gamma ]\!]}([\![A]\!])$ 
        \item If $\Gamma \vdash \alpha : M \simeq_A N$, then $\alpha \in RedEquiv_{[\![A ]\!]}^{[\![\Gamma ]\!]}(M, N)$
    \end{itemize}
\end{theorem}

\begin{proof}
    The entire proof can be found in \cite[(Pages 8-10)]{Licata12}. It proves the result by induction on type 
    derivations, defining the interpretation compositionally for each type constructor.
\end{proof}

Now we can prove the canonicity of 2DTT as a corollary:

\begin{proof}
    Given a theorem above, assume $\bullet \vdash M : Boolean$. Then $M \in RedTm^{[\![\bullet ]\!]}([\![Boolean]\!])$, so 
    $M \in RedTm^{[\![\bullet ]\!]}(const(\mathbbm{2}))$. By definition $id \in Ob(\bullet)$, so $\llangle M[id]\rrangle \in 
    Ob(const(\mathbbm{2})(\llangle id \rrangle))$, so $const(\mathbb{2})(\llangle id \rrangle)) = \mathbbm{2}$.
    By definition there are only  2 objects of $\mathbbm{2}$, so $M \equiv true$ or $M \equiv false$.
\end{proof}

\subsection{Limitations}

It is impossible to extend the methodology of 2TT to further dimensions due to problems with coherence 
at identity types beyond the first level. A different judgemental framework would need to be utilised.

As such, limitation of a 2TT is its two-dimensionality. As we discussed above, propositional truncations 
of type theory do not have decidable type checking and judgemental equality as a result of "cutting off" 
computational content of higher equalities. The theory does, however, retain decidability of judgemental 
equality and type checking for types that do not include identity types in their constructions, which is still
a lot of useful types. In fact, if it was possible to extend the judgemental presentation to further dimensions,
more types (ones using first/second/n'th level of identity types) would become decidable, achieving full 
decidability in the limit.

Furthermore, a canonicity in 2TT is achieved by equivalence relations, rather than one-directional computation
rules such as function application. As such, although the theory enjoys canonicity, there is no effective algorithm 
to compute canonical forms. Further work needs to be done to explore a possibility of extending 2TT 
with a one-directional operational evaluation strategy that would ensure that all terms compute to a 
canonical form by a terminating algorithm.

\section{Conclusion}
Dependant type theory is an important recent development for both mathematics and computer programming. On a 
mathematical side of things, HoTT is a promising candidate for a new foundational language, able to formalise 
constructive mathematics in a verifiable way. On a computer scientific side of things, dependant types represent 
a powerful tool for formally verifiable systems and proof assistants.

A univalence axiom is a foundational idea for both fields, being invaluable for representing mathematics in a 
more natural way, studying homotopy theory inside new mathematical foundations. For programming this axiom is 
important as a generic program that is able to lift equivalences across type families. The axiom, however, as 
of yet lacks computational content in the most general presentation of HoTT. Although some attempts have been
made to create simpler theories that attain a judgemental presentation of univalence, more work is needed to 
achieve a computationally-relevant version of univalence in a general setting.

\nocite{*}
\bibliography{mybib}{}
\bibliographystyle{plain}

\end{document}